%% file: brief_tac.tex
\documentclass{IEEEtran}

\input{preambleLettersDef.tex}

\title{Transfer Learning for LQR Control}

\usepackage{amsmath, amsthm}
\usepackage{bm} 
\usepackage{cite}
\usepackage{bm}
\usepackage{amsmath,amssymb,amsfonts}
\usepackage{algorithmic}
\usepackage{textcomp}
\usepackage{graphicx,color}
\usepackage{mathrsfs}
\usepackage{comment}
\usepackage{float}
\usepackage[ruled]{algorithm2e}
\usepackage{soul}
\usepackage{url}
\usepackage{color}
\usepackage{dsfont}
 \usepackage{bbm}
\usepackage{booktabs}
\usepackage{array}
\usepackage[table]{xcolor} 
\usepackage{yfonts}
\usepackage[normalem]{ulem}
\newcommand{\stkout}[1]{\ifmmode\text{\sout{\ensuremath{#1}}}\else\sout{#1}\fi}
\usepackage{arydshln}
\usepackage{multicol}
\usepackage{amsmath}

\newtheorem{theorem}{Theorem}[section]
\newtheorem{lemma}[theorem]{Lemma}
\newtheorem{definition}{Definition}

\newtheorem{remark}{Remark}

\newcommand{\subscr}[2]{{#1}_{\textup{#2}}}

\usepackage{tabularx}
\usepackage{multirow}
\usepackage{makecell}

\newcommand\aamsout{\bgroup\markoverwith{\textcolor{violet}{\rule[0.5ex]{2pt}{1pt}}}\ULon}

\newcommand{\T}{\mathsf{T}}

\newcommand{\mc}{\mathcal}

\newtheorem{Assumption}{Assumption}

\newcommand{\argmin}[2] {\mathrm{arg}\min_{#1}#2}

\DeclareSymbolFont{bbold}{U}{bbold}{m}{n}
\DeclareSymbolFontAlphabet{\mathbbold}{bbold}

\newcommand\oprocendsymbol{\hbox{$\square$}}
\newcommand\oprocend{\relax\ifmmode\else\unskip\hfill\fi\oprocendsymbol}

\newcommand*{\QEDA}{\hfill\ensuremath{\blacksquare}}%

\makeatletter
\let\NAT@parse\undefined
\makeatother

\usepackage[urlcolor=blue,linkcolor=blue, citecolor=blue]{hyperref}

\author{Taosha Guo and Fabio Pasqualetti 
\thanks{This material is based upon work supported in part by ARO award W911NF-24-1-0228.
T. Guo and F. Pasqualetti are with the
Department of Mechanical Engineering at the 
University of California, Riverside, Riverside, CA, 92521, USA. Emails: \href{mailto:tguo023@ucr.edu}{\{\texttt{tguo023}},
\href{mailto:fabiopas@ucr.edu}{\texttt{fabiopas\}@ucr.edu}},
}
}

\begin{document}

\graphicspath{{img/}}
\maketitle 
\begin{abstract}
  In this paper, we study a transfer learning framework for Linear
  Quadratic Regulator (LQR) control, where (i) the dynamics of the
  system of interest (target system) are unknown and only a short
  trajectory of impulse responses from the target system is provided,
  and (ii) impulse responses are available from $N$ source systems
  with different dynamics. We show that the LQR controller can be
  learned from a sufficiently long trajectory of impulse
  responses. Further, a transferable mode set can be identified using
  the available data from source systems and the target system,
  enabling the reconstruction of the target system's impulse responses
  for controller design.  
  By leveraging data from source systems, we show that the sample complexity for synthesizing the LQR controller can be reduced by $50 \%$. Algorithms and numerical examples are provided to demonstrate the implementation of the proposed transfer control framework.
  \end{abstract}
\maketitle

\begin{IEEEkeywords}
Optimal control, Transfer learning, Data-driven control
\end{IEEEkeywords}
\section{Introduction}
  In real-world applications, we often encounter complex systems with unknown dynamics, where substantial amounts of training data are used to design  control policies. 
  However, in many practical scenarios, obtaining sufficient data from the system of interest (target system) is infeasible due to limitations such as high cost, incomplete observations, or limited access.  For instance, when  operating Mars rovers in an unknown  environment, designing the control strategy becomes challenging because real-time data collection is limited by the significant cost and delays associated with communication.

 However, in some cases, we may   have access to data from  similar systems (source systems) that share certain characteristics with the target system.  For instance,  data collected from previous Mars exploration missions may provide valuable insights into the dynamics of the current mission's environment. This raises a natural question: how can we leverage the information from  source systems to design control strategies for the target system? 

In this paper we investigate the LQR problem \cite{KZ-JCD-KG:96}, which seeks a control policy for a dynamical system that minimizes a quadratic function of the output and input. Specifically, we focus on linear-time-invariant (LTI) systems with unknown dynamics. We assume access to a short trajectory of data collected from the unknown target system, where the limited data is insufficient to identify the target system's model, or to learn the optimal control policy using existing  direct-data-driven techniques.  However, we have abundant data from $N$ numbers of source systems that share certain characteristics with the target system. 
The question we answer is:   Can data from source systems be leveraged to learn the optimal controller for the target system?
 We demonstrate  that the limited  data from the target system can be utilized  to learn a transferable eigenvalue set from source systems, such that the amount of data required for controller synthesis can be reduced.
  
\noindent \textbf{Related Work.}
In recent years, data-driven control has gained significant attention due to its ability to design control strategies  directly from data, bypassing  the need for explicit models. 
Most of these approaches,  however, require a sufficiently large amount of open-loop data to learn the optimal controller. Additionally, the collected data also needs to satisfy persistent excitation conditions, as demonstrated   in \cite{BR:18}, \cite{KZ-BH-TB:20},   \cite{IM-FD:21}, \cite{GS-AB-CL-LC:18} and \cite{FC-GB-FP:22}. Different from these works, we consider a transfer learning framework where the optimal control policy is computed using  information transferred from source system data, and a small amount of impulse response data collected from the target system. 
 Transfer learning for control systems has also recently gained popularity. For example, 
 \cite{LX-LY-GC-SS:22} identifies the dynamics of an unknown system from samples generated by a similar system;
 \cite{TG-AAAM-VK-FP:23} introduces an imitation and transfer learning framework for LQG control.  Similarly, \cite{RB-VB-FB-SF:24} and  \cite{LL-CDP-PT-NM:23} leverage prior knowledge from similar systems  to reduce experimental workload and the amount of data required for controller design. This paper differs from  \cite{TG-AAAM-VK-FP:23} in that the unknown target system we consider has different dynamics from those of the source systems. Furthermore, unlike \cite{LX-LY-GC-SS:22}, \cite{RB-VB-FB-SF:24} and \cite{LL-CDP-PT-NM:23}, our methods do not impose constraints on the similarity between the target and source systems. Instead, we use the limited data collected from the target system to identify a transferable eigenvalue set from the source systems, 
 enabling efficient controller synthesis with less data.
 
\noindent \textbf{Contribution.}
The main contributions of this
paper are as follows: First, we  show that the optimal control inputs can be computed using a  gain matrix  and a finite window of past inputs and outputs, we  derive a closed-form solution to compute the gain matrix directly from impulse response data. 
Second, we formalize the transfer learning framework  for optimal controller design,  where the impulse responses from $N$ source systems are available alongside limited impulse response data from the target system. \textcolor{black}{By exploiting data from source systems, we show that the sample complexity for LQR controller design  can be reduced by $50\%$. We also provide algorithms and numerical experiments that implement the proposed transfer-control framework.
}

\noindent \textbf{Organization of the paper.} 
The paper is organized as follows. 
Section \ref{Sec: prob for} describes the problem formulation and some useful background notions. Section \ref{Section: LQR design} contains 
the results to compute the data-driven controller in a closed-form using the impulse response data.
Section \ref{Section: transfer learning} demonstrates the results  and algorithms of transfer learning control. Section \ref{Sec: numerical example}  and Section \ref{conclusion} contain  numerical examples and  conclusions.
\section{Problem formulation and preliminary notions}{\label{Sec: prob for}}
Consider the discrete-time, linear, time-invariant system
\begin{equation}\label{eq: system}
\begin{aligned}
    x(t+1) &= A x(t) +B u(t) ,\\
        y(t)  & = Cx(t
        ),
\end{aligned} 
\end{equation} 
where $x(t)\in \mathbb{R}^n$ denotes the state, $u(t)\in \mathbb{R}^m$
the control input, and $y(t)\in \mathbb{R}^l$ the measured output. We
assume that \eqref{eq: system} is controllable and observable. The
finite-horizon LQR problem seeks to find the input sequence
$\{u(t)\}_{t=0}^{T-1}$ that minimizes the cost
\begin{align}\label{eq: LQR cost}
  J = 
   y(T)^{\T}Q y(T)+
\sum_{t=0}^{T-1}y(t)^{\T}Q y(t) +
  u(t)^{\T} R u(t)  , 
\end{align}
where $Q\succeq 0$, $R\succ 0$ are weight matrices and $T$ is the
control horizon. 
If the system matrices $A,B$, and $C$ are given, solutions to this
problem are well known. For instance, the control sequence that
minimizes \eqref{eq: LQR cost} can be obtained through a time-varying
feedback of the state, that is,
\begin{align}{\label{eq: optimal input}} u(t) = - K_t
  x(t),
\end{align}
where
\begin{align}{\label{eq:Kt}}
	K_t = (R+B^{\T}P_{t+1}B)^{-1}B^{\T}P_{t+1}A,
\end{align}
and $P_t$ is computed backwards in time using the following recursion
\begin{equation}{\label{eq: Riccati Pt}}
\begin{aligned}
  P_t  = 	& A^{\T} P_{t+1} A - A^{\T}P_{t+1}B(R+B^{\T}P_{t+1}B)^{-1}\\
  &B^{\T}P_{t+1}A+C^{\T} Q C, \\
\end{aligned}
\end{equation}
with the terminal condition $P_T = C^{T}QC$. When $T = \infty$, the time
varying gain $K_t$ in \eqref{eq: optimal input} converges to the
static gain
\begin{align}{\label{eq:K}}
  K^{*} = (R+B^{\T}P^{*}B)^{-1}B^{\T}P^{*}A ,
\end{align}
where $P^{*}$ is the  (unique) positive definite solution to the discrete time algebraic Riccati equation 
\begin{equation}{\label{eq: Riccati P}}
\begin{aligned}
P^{*}  = & A^{\T} P^{*} A - A^{\T}P^{*}B(R+B^{\T}P^{*}B)^{-1}
\\ & B^{\T}P^{*}A+C^{\T} Q C.
\end{aligned} 
\end{equation}

The above procedures to compute a solution to the LQR problem, and in fact most known model-based solutions \cite{KZ-JCD-KG:96}, requires
the full knowledge of the system dynamics. Data-driven alternatives for the solution of the LQR problem
also exist, which do not require knowledge of the system dynamics yet
necessitate a sufficient, often large, amount of training data, see,
e.g., \cite{KZ-BH-TB:20}, \cite{IM-FD:21}, \cite{GS-AB-CL-LC:18} and
\cite{FC-GB-FP:22}. In this paper we investigate whether it is
possible to reduce the amount of data required to solve the LQR
problem in a data-driven setting by leveraging data from systems with different, yet similar, dynamics.
 In particular, we let the
dynamics of the target system \eqref{eq: system}, denoted by $\Scal_0$ be unknown, yet we
assume the availability of the impulse response 
$M(1:T_0):= \{M(1), M(2), \cdots, M(T_0)\}$ with $T_0 \leq 2n$, so that the existing system identification methods \cite{SO-NO:22} are unable to identify its dynamics. Additionally, we assume the availability of the impulse response  $\{M_i(1:T_1)\}_{i=1}^N$
 from $N$ numbers of source systems, denoted by $\{\Scal_i\}_{i=1}^N$. Our objective  is to learn
the optimal controller for $\Scal_0$ using only $M(1:T_0)$ and $\{M_i(1:T_1)\}_{i=1}^N$, and to characterize how large  $T_0$, $N$ and $T_1$ should be to solve the LQR problem. 
\section{Data-driven optimal controller design}{\label{Section: LQR design}}
 In this section, we present our first result, that is, given the weight matrices $Q$ and $R$, compute  the optimal control sequences $\{u(t)\}_{t=0}^{T-1}$ that minimizes \eqref{eq: LQR cost} by using impulse responses $M(1,T)$  collected from an unknown system. 
 To begin with,  following the procedures in \cite{AAAM-VK-VK-FP:22}, we rewrite the optimal inputs generated by the dynamic controller \eqref{eq: optimal input}  using a controller gain and a finite window of past inputs and outputs
 \begin{align}{\label{eq: static controller}}
	u(t) =  
	\subscr{K}{LQR}^{t}
	\begin{bmatrix}
		U_n(t) \\ Y_n(t)
	\end{bmatrix},
\end{align}
 where
\begin{align}\label{eq: K LQR t}
	\subscr{K}{LQR}^{t} = K^{t}\begin{bmatrix}
		Fu-A^n\Ocal^{\dagger}\bm{S}_{T-n+1} & A^n\Ocal^{\dagger}
	\end{bmatrix}
\end{align} and $U_n(t)$, $Y_n(t)$ are constructed using $n$ samples of past input and output data
\begin{align}
	U_n(t) = \begin{bmatrix}
		u(t-n) \\
		\vdots \\
		u(t-1)
	\end{bmatrix}, \;
	Y_n(t) = \begin{bmatrix}
		y(t-n) \\
		\vdots \\
		y(t-1)
	\end{bmatrix},
\end{align}
and  $Fu=\begin{bmatrix}
		A^{n-1}B & \cdots & AB & B
	\end{bmatrix}$ and $\Ocal = \begin{bmatrix}
		C^{\T} & (CA)^{\T} & \cdots & (CA^{n-1})^{\T}
	\end{bmatrix}^{\T}$.  
 In the next theorem we show that $\subscr{K}{LQR}^{t}$ can be computed using the weight matrices $Q$ and $R$ and the impulse responses $M(1:T)$.
\begin{theorem}{\label{thm: Output LQR controller}}{\bf{(Data-driven optimal controller)}}
The controller gain $\subscr{K}{LQR}^{t}$ in 
\eqref{eq: static controller} has the following alternative expression:
 \begin{equation}{\label{eq: Output feedback LQR}}
\begin{aligned}
	\subscr{K}{LQR}^{t} & = - 
	\left[ R+\bm{M}_{t+1}^\T (\bm{Q}_{t+1}^{-1}+ \bm{S}_{t+1} \bm{R}_{t+1}^{-1} \bm{S}^{\T}_{t+1})\bm{M}_{t+1} \right ]^{-1} \\
	& \bm{M}^{\T}_{t+1}(\bm{Q}_{t+1}^{-1} 
	+\bm{S}_{t+1}\bm{R}_{t+1}^{-1}\bm{S}^{\T}_{t+1})
	\begin{bmatrix}
	 E- F\bm{S}_{T-n+1} & F 
	\end{bmatrix},
\end{aligned} 
 \end{equation}
where $\bm{Q}_t = \emph{diag}(Q, \dots,Q)$
and $\bm{R}_t = \emph{diag}(R, \dots,R)$ respectively contain $T-t+1$ diagonal blocks; 
 $\bm{M}_t$ and $E$ are constructed using the impulse responses $M(1:T)$ collected from the unknown system:
{\footnotesize{
\begin{align*}
	\bm{M}_t&  = 
	\begin{bmatrix}
		M(1) \\
		\vdots \\
		M(T-t+2)
	\end{bmatrix}, 
	\\
	E  &= 
	\begin{bmatrix}
		M(n+1) & \cdots & M(2) \\
		\vdots & \ddots & \vdots \\
		M(T-t+n+2) & \cdots & M(T-t+3)
	\end{bmatrix}, \\
 \bm{S}_t  &= \begin{bmatrix} 0  \\
 M(1) & 0 \\
\vdots &  & \ddots  \\
M(T-t) & M(T-t-1) &  \cdots & M(1) & 0
 \end{bmatrix} .
\end{align*} }}
\end{theorem}
We postpone the proof of Theorem \ref{thm: Output LQR controller} to the Appendix.
We use different expressions to compute matrix $F$ in  \eqref{eq: Output feedback LQR} for the single output ($l=1$) and multiple output ($l>1$) cases. We first discuss the case for $l=1$. 
 Let $\{\alpha_1, \cdots, \alpha_n\}$ represent the coefficients of the characteristic polynomial of $A$ in \eqref{eq: system}. According to the Cayley-Hamilton Theorem we have
  \begin{align}{\label{eq: Calay th}}
  A^n +\alpha_1I+ \alpha_2 A+\cdots+\alpha_n A^{n-1} = 0.
   \end{align}
   Define
  $\bm{\alpha}=\begin{bmatrix}
  	\alpha_1 & \cdots & \alpha_n
  \end{bmatrix}$ and let $M^j(t)$ with $j\in \{1, \cdots,m\}$ be the $j$-th column of $M(t)$, then  we obtain $\bm{\alpha}$ by
  \begin{align}{\label{eq: estimate alpha}}
  \bm{\alpha} 
  =
  \begin{bmatrix} 
  - M^j(1) & \cdots & M^j(n) \\
  \vdots &   \ddots &   \vdots \\
   M^j(n) & \cdots & M^j(2n-1) 
  \end{bmatrix}^{\dagger}
  \begin{bmatrix} 
  	M^j(n+1) \\ \vdots \\ M^j(2n)
  \end{bmatrix}.
    \end{align}
     For $-n\leq i<0$, let  $F_i $  be  the $(n+i+1)$th row  of $I_n$.  For $i \geq 0$, we calculate each $F_i$ recursively by  $F_i = -\bm{\alpha}
 \begin{bmatrix}
 	F_{i-n+1}^{\T} &
 	\cdots &
 	F_{i-1}^{\T}
 \end{bmatrix}^{\T}$. Then, $F = 
 \begin{bmatrix}
 F_1^{\T} & \cdots & F_{T-t+2}^{\T} 
 \end{bmatrix}^{\T}$.
The next remark gives the procedures to compute $F$ in \eqref{eq: Output feedback LQR} for the  case when $l>1$.
\begin{remark}{\bf \emph{(Computing  $F$ in \eqref{eq: Output feedback LQR} for $l>1$)}}
For the multiple output case we first  take the QR factorization of the following  data matrix
\begin{align}{\label{eq: QR decomposition}}
	LH = \begin{bmatrix}
		M(1) & \cdots & M(k) \\
		\vdots & \ddots & 	\vdots\\
		M(n) & \cdots &  M(k+n-1)
	\end{bmatrix},
\end{align}
 where $L \in \mathbb{R}^{nl\times nl}$ and $H \in \mathbb{R}^{nl \times mk}$, and 
 $k$ should satisfy $mk \geq n$. Let $\tilde{O}$ be the first $n$ columns of $L$ and define $\tilde{L} =\tilde{\Ocal}\tilde{\Ocal}^{\dagger} $, where $\tilde{L} = \begin{bmatrix}
	\tilde{L}_1^{\T} & \cdots &
	\tilde{L}_n^{\T} 
\end{bmatrix}^{\T} \in \mathbb{R}^{nl \times n}$ with $\tilde{L}_j \in \mathbb{R}^{l\times n} $. Similar to the single input case, for $1-n<i<0$, we define $F_i = \tilde{L}_{n+i+1}$. For  $i\geq 0$ and we compute $F_i$ recursively  by $ F_i =  - \bm{\alpha}\otimes I_l \begin{bmatrix}
	F_{i-n+1}^{\T} & \cdots & F_{i-1}^{\T}
\end{bmatrix}^{\T}$. 
\end{remark}
As $T \to \infty$,  $\subscr{K}{LQR}^{t}$ converges to 
\begin{align}{\label{eq: K LQR star}}
	\subscr{K}{LQR}^{*} = K^{*}\begin{bmatrix}
		Fu-A^n\Ocal^{\dagger}\bm{S}_{T-n+1} & A^n\Ocal^{\dagger}
	\end{bmatrix}.
	\end{align} Next we characterize the convergence from   $\subscr{K}{LQR}^{t}$  to   $\subscr{K}{LQR}^{*}$.
	\begin{lemma}{\label{lemma: convergence ana}}{\bf \emph{(Convergence  from  $\subscr{K}{LQR}^{t}$ to $\subscr{K}{LQR}^{*}$ }}) For a given $T$, let $\subscr{K}{LQR}^{t}$ and $\subscr{K^*}{LQR}$ be as in \eqref{eq: K LQR t} and 
\eqref{eq: K LQR star}, respectively. Then,
\begin{align}
	\|\subscr{K}{LQR}^{T-t} - \subscr{K}{LQR}^{*}\|\leq c_1 c_2 \mu^{t},
	\end{align}
	for $0<\mu<1$ and $c_1, c_2 \in \mathbb{R}^{+}$.
	\end{lemma} We postpone the proof of Lemma \ref{lemma: convergence ana} to the Appendix. 
\section{Transfer learning based optimal control}{\label{Section: transfer learning}
      \textcolor{black}{The optimal controller introduced in Theorem \ref{thm: Output LQR controller}  generates the optimal control sequences from a finite window of past inputs and outputs, 
  where the output feedback LQR gain $\subscr{K}{LQR}^{t}$ can be obtained from a sufficiently large amount of impulse response data. 
Recall that we consider the setting in which data collected from the target system $\mathcal{S}_0$ are insufficient to directly design the optimal controller or identify the dynamics, whereas trajectory data from $N$ source systems, denoted by  $\{ \Scal_i\}_{i=1}^{N}$, are available.} 
  Then finding $\subscr{K}{LQR}^{t}$ for the unknown target system $S_0$
  reduces to approximating the impulse response $M(T)$ of $\Scal_0$ for any $T$,   using $\{M_i(1:T_1)\}_{i=1}^N$ collected from 
  $\{ \Scal_i\}_{i=1}^{N}$ and $M(1:T_0)$ from $\Scal_0$.
  
\subsection{Estimating impulse response for target system}
To approach this problem we first show that if the modes (eigenvalues) of  $\Scal_0$ can be approximated from the source systems, the impulse responses $M(T)$ for any $T$ can be computed  for $\Scal_0$ by sampling a short impulse response trajectory $M(1:T_0)$, which is insufficient to identify the dynamics of $\Scal_0$.
    To begin with, we present  the following lemma to show how to estimate $M(T)$ if the system modes and a short trajectory of impulse response data $M(1:T_0)$ are given. We first make the following technical assumption and definition.  
\begin{Assumption}{\label{Assumpt: controllarbility}}{\bf \emph{(Observability and controllability)}} The target system $\Scal_0$ and source systems $\{\Scal_{i}\}_{i=1}^{N}$ have the same state, input and output dimensions, and are all controllable and observable. 
\end{Assumption}
\begin{definition}{\label{definition}}
	Let $p = \begin{bmatrix}
		p(1) & \cdots & p(n)
	\end{bmatrix} \in \mathbb{R}^{1 \times n}$, we define the following operator $\mathcal{H}(p,t) \in \mathbb{R}^{lm \times lmn}$ that uses vector $p$ to construct a block diagonal matrix with $l \times m$ numbers of blocks:
	{ 
	\footnotesize{
		\begin{align}
		\mathcal{H}(p,t) = 
		\begin{bmatrix}
		\begin{bmatrix}
		p(1)^t & \cdots & p(n)^t 
		\end{bmatrix} \\
		& & \ddots & \\
		& &&\begin{bmatrix}
		p(1)^t & \cdots & p(n)^t 
		\end{bmatrix}
			\end{bmatrix}.
	\end{align} }}
\end{definition}

\begin{lemma}{\label{Lemma: Decompose Markov para}}{\bf \emph{(System modes and impulse responses)}}
Let $\Tcal \in \mathbb{R}^{n \times n}$ be an invertible matrix such that $\Tcal^{-1}A\Tcal = \emph{diag}(\sigma_1,\cdots,\sigma_n)$ and define $\tilde{C} = C\Tcal$ and $\tilde{B} = \Tcal^{-1}B$. 
 Let $\Sigma := \{\sigma_1, \cdots, \sigma_n\}$ denote the modes  of $\Scal_0$. Let $\tilde{M}(t):= \emph{blkdiag} (M_1(t),\cdots,M_m(t))$, where $M_i(t)$ represents the $i$-th column of $M(t)$. Then $\tilde{M}(t)$ satisfies {
 \begin{align}{\label{eq: Decompose Markov para}}
\tilde{M}(t)
=
\mathcal{H}(\Sigma,t-1)
\underbrace{
\begin{bmatrix}
	W_1\\
	& \ddots  \\
	 & & W_m	\end{bmatrix}}_{\bm{W}},
\end{align}} where $\mathcal{H}$ is defined in Definition \ref{definition} and
 $ W_i: = 
 \begin{bmatrix}
 \tilde{b}_i^{\T} \emph{diag}(\tilde{c}_1) & \cdots & \tilde{b}_i^{\T} \emph{diag}(\tilde{c}_l)
 \end{bmatrix}^{\T}
$ with $\tilde{c}_i$ being the $i$-th row of $\tilde{C}$ and $\tilde{b}_i$  the $i$-th column of $\tilde{B}$ respectively. 
\end{lemma} We postpone the proof of Lemma \ref{Lemma: Decompose Markov para} to the Appendix.  From \eqref{eq: Decompose Markov para} we have
\begin{align*}
\underbrace{
\begin{bmatrix}
	\tilde{M}(2) \\ \vdots \\ \tilde{M}(T_0+1)
\end{bmatrix}}_{\tilde{M}(2:T_0+1)} =
\underbrace{
\begin{bmatrix}
	\mathcal{H}(\Sigma,1)  \\  \vdots \\ \mathcal{H}(\Sigma,T_0) 
	\end{bmatrix}}_{{H}_{1:T_0}} \bm{W}.
\end{align*}
If ${H}_{1:T_0}$ is of full column rank, $\bm{W}$ can be simply  obtained by $ \bm{W}= H_{1:T_0}^{\dagger}\tilde{M}(2:T_0+1)$. Accordingly,
\begin{align}{\label{eq: stack matrix}}
	M(T) = 
	\begin{bmatrix}
		I_l & \cdots & I_l 
	\end{bmatrix}
\mathcal{H}(\Sigma,T-1)H_{1:T_0}^{\dagger} \tilde{M}(2:T_0+1).
\end{align}
\textcolor{black}{ 
The above results  demonstrate that $M(T)$ can be predicted for any $T$ using a short trajectory $M(1:T_0)$, provided that the modes of $\mathcal{S}_0$ are obtained}.
Next we characterize how large $T_0$ should be to compute $M(T)$ using \eqref{eq: stack matrix}.
\begin{lemma}({\bf{Sample complexity analysis}}) Let $T_0 \geq n$. Then, $M(T)$ for the unknown target system $\Scal_0$ can be reconstructed  from the modes $\Sigma$ and $M(1:T_0+1)$  by \eqref{eq: stack matrix}.
\end{lemma}
\begin{proof}
	 For the block diagonal matrix $H_{1:T_0} \in  \mathbb{R}^{lmt \times lmn}$ in \eqref{eq: stack matrix} to be of full column rank, we need $lmT_0 \geq lmn$. Therefore  $T_0 \geq n$. 
\end{proof}

\textcolor{black}{Conventionally, to compute $M(T)$ for an unknown target system $\Scal_0$ we  (i) first  identify the system matrices $\{A,B,C\}$ using the  Ho-Kalman algorithm \cite{SO-NO:22} from  $M(1:T_0+1)$;  (ii) then evaluate $M(T) = CA^{T-1}B$ for the desired horizon $T$. In particular, the Ho-Kalman  algorithm requires $T_0 +1\geq 2n$ to identify $\{A,B,C\}$ for an $n$-th order system, whereas when system modes $\Sigma$ are obtained, \eqref{eq: stack matrix} only requires   $T_0 +1 \geq n$,
\textcolor{black}{thereby reducing the sample complexity by $50\%$.} }
In the next  section we show how to  use $\{M_i(1:T_1)\}_{i=1}^N$ collected from the source system to 
learn $\Sigma$. 
\subsection{Mode transfer from  source systems}

Using \eqref{eq: estimate alpha} we can obtain the corresponding characteristic coefficients   $\bm{\alpha}$ for each source system, as defined in 
\eqref{eq: Calay th}
 , by collecting impulse response data $\{M_i(1:2n)\}_{i=1}^N$, then the roots of
\begin{align}{\label{eq: roots}}
\lambda^n +\alpha_n \lambda ^{n-1}+ \cdots +\alpha_2 \lambda +\alpha_1  = 0
\end{align} are the modes corresponding to  each source system. 
Let $\Lambda = \{\lambda_1, \dots, \lambda_k\}$ denotes the set of the distinct modes among  all source systems, solved by \eqref{eq: estimate alpha} and \eqref{eq: roots} using  $\{M_i(1:2n)\}_{i=1}^N$. We refer to $\Lambda$ as the mode dictionary for the rest of the paper.
  Next we partition $\Lambda$ into K numbers of subsets denoted by 
  \begin{align}{\label{eq: searching }}
   \{\Lambda_j\}_{j=1}^{K} := \{\lambda_{j1}, \cdots, \lambda_{j n} \}_{i=j}^{K},      
   \end{align}
   where $\Lambda_j \subset \Lambda$ and
 each subset $\Lambda_j$ contains $n$  elements from $\Lambda$, and $K = \binom{k}{n} = \frac{k!}{n!(k-n)!}$ is the total number of all possible combinations of $\Lambda_j$ constructed using the $n$ distinct elements from $\Lambda$.
 We present the next lemma to demonstrate how to estimate $\Sigma$ from  $\Lambda$ by using  the first $n+1$ samples of impulse responses $M(1:n+1)$ collected from the unknown target system $\Scal_0$. 
\begin{lemma}{\bf \emph{(Relations between $\Sigma$ and impulse responses)}}{\label{lemma: impulse response and modes}} Let $M(1:n+1)$ denotes the impulse responses collected from  $S_0$, and $ \Sigma :=   \{\sigma_1,  \cdots,  \sigma_n \}
      $ represents the modes of $\mathcal{S}_0$. Then,   
     \begin{align}{\label{eq: M and Sigma}}
     \begin{bmatrix}
     		M(n+1) & M(1) & \cdots & M(n)
     	\end{bmatrix}
     	\begin{bmatrix}
     	I_m \\
     	\bm{\alpha}^{\T} \otimes I_{m}
     	\end{bmatrix}=0,
     \end{align}
     where $\bm{\alpha}$ is defined in \eqref{eq: Calay th} and satisfies $\lambda^n +\alpha_n A^{n-1}+ \cdots +\alpha_2 \lambda +\alpha_1 = (\lambda - \sigma_1)(\lambda - \sigma_2)\cdots(\lambda - \sigma_n)$.
   \end{lemma}
   We postpone the proof of Lemma \ref{lemma: impulse response and modes} to the Appendix.  Lemma \ref{lemma: impulse response and modes} shows that $M(1:n+1)$ and $\Sigma$ should always satisfy \eqref{eq: M and Sigma},  therefore we can use this property to learn $\Sigma$ from the mode dictionary $\Lambda$ by 
    \begin{align}\label{estimation}
    {\footnotesize
\bm{\hat{\alpha}}= 
\argmin{\bm{\alpha}_j} 
\underbrace{
  \left  \|
 \begin{bmatrix}
     		M(n+1) & M(1) & \cdots & M(n)
     	\end{bmatrix}
     		\begin{bmatrix}
     	I_m \\
     	\bm{\alpha}_j^{\T} \otimes I_{m}
     	\end{bmatrix}
     		 \right \|}_{Z(\bm{\alpha}_j)},}
	  \end{align} 
	  where $\bm{\alpha}_j$ denotes the coefficient vector corresponding to the set $\Lambda_j$.
	  \eqref{estimation} also provides a metric to measure the similarity between the source systems  and the target system in terms of their eigenvalues, which is captured by their impulse responses $M(1:n+1)$ and $\{M_i(1:2n)\}_{i=1}^{N}$. \textcolor{black}{In particular, if $Z(\hat{\bm{\alpha}}) = 0$, 
	    then $\Sigma \subset \Lambda$, meaning that the exact modes of the target system are accurately retrieved from \eqref{estimation}}. In practice, we can use \eqref{estimation} to select sources systems that return the smallest $Z(\hat{\bm{\alpha}})$. 
	  	 		  We summarize the above  procedures of learning   $\Sigma$ into Algorithm \ref{Al: mode detection}.
		  	\begin{algorithm}[!ht]
\caption{Modes transfer from source systems}
\label{Al: mode detection}
	\KwIn{
   $\{M_i(1:2n)\}_{i=1}^N$ and $M(1:n+1)$ \;\\ 
	}
	\textbf{Step 1:} {Construct the $\Lambda$ by \eqref{eq: estimate alpha} and \eqref{eq: roots}\\ 
	$ \hspace{2em} \Lambda \leftarrow   \{M_i(1:2n)\}_{i=1}^N$} 
	
	\textbf{Step 2:} Construct the $K$ combinations of modes
	\\ 
	 $ \; \; \{\Lambda_j\}_{j=1}^{K} \leftarrow   \Lambda$  
	
		\textbf{Step 3:} Compute the coefficient vectors $\{ \bm{\alpha}_j \}_{j=1}^{K}$
		\\ 
	$\{\bm{\alpha}\}_{j =1}^{K} \leftarrow \{\Lambda_j\}_{j=1}^{K}$  
	
	\textbf{Step 4:}
	Find the $\hat{\bm{\alpha}}$ that minimizes \eqref{estimation} \\
	${\hat{\bm{\alpha}}} \leftarrow \{\bm{\alpha}_j\}_{j=1}^{K}$ \\
	\KwOut{ The mode set $\hat{\Sigma}$ corresponding to $\hat{\bm{\alpha}}$}
\end{algorithm}
	\section{Numerical examples and comparison}
	\label{Sec: numerical example}

	We now analyze the results from Section \ref{Section: LQR design} and \ref{Section: transfer learning}  by means of numerical simulations. Suppose we have two source systems $\Scal_1 $ and $ \Scal_2$ with the following dynamics (unknown):
	{{
\begin{align*}
		A_1 & = 
		\begin{bmatrix}
    0.41  &  1.56  &  -1.59 \\
    0.06  &  1.34  & -1.25 \\
   -0.30  &  1.24  & -1.14
     \end{bmatrix},
   &   B_1 & = \begin{bmatrix}
   -0.47 \\
   -0.81 \\
   1.00
    \end{bmatrix}, \\
    C_1 &  = 
    \begin{bmatrix}
    	1.80 & -2.75 & 0.76
    \end{bmatrix}
 \\
		A_2 & =  \begin{bmatrix} 
    1.42 &   -2.91 &    3.58 \\
    1.24 &   -3.57 &    5.12 \\
    0.55 &   -2.37 &    3.74 		\end{bmatrix},
   & B_2  &= \begin{bmatrix}
    1.17 \\
    1.37 \\
    0.75
    \end{bmatrix}, \\
    C_2  &= 
    \begin{bmatrix}
 -0.46 &   0.03  &  1.36
     \end{bmatrix}.
    	\end{align*}}}
		  \noindent Then we collect $M_1(1:2n)$ and $M_2(1:2n)$ and follow \eqref{eq: estimate alpha}  the mode dictionary $\Lambda$ is obtained as: 
		\begin{align}
		\Lambda = 
\{
0.11,-0.52,1.02,0.21,0.36
 \}.
 \end{align}
 
\textcolor{black}{In the following, we analyze two scenarios. In the first scenario, all modes of the unknown target system are contained in the mode dictionary $\Lambda$, i.e., $\Sigma \subset \Lambda$. In the second scenario, the target system modes are not fully contained in the mode dictionary, that is, $\Sigma \not \subset \Lambda$.}

\subsection{Modes are fully contained in the dictionary: $\Sigma \subset \Lambda$}{\label{subsection A}}
Our  target system  $S_0$ is open-loop stable and has the following (unknown) dynamics:
{{
\begin{align*}
A & =  
\begin{bmatrix}
 1.93  &  -0.87 &  -0.27 \\
 1.11  &  -0.07 &  -0.36 \\
 -1.17  &  2.30  & -0.89
\end{bmatrix},
B =
\begin{bmatrix}
 -0.07 \\
    0.32\\
   -0.01\\
\end{bmatrix}, \\
C & =
\begin{bmatrix}
 -0.04 &  -0.32  &  2.33
\end{bmatrix}.
 \end{align*}}}
\!\!The impulse responses   
$M(1:4)$
are collected from the target system $\Scal_0$. 
\textcolor{black}{Following Algorithm \ref{Al: mode detection} we obtain $\min Z(\hat{\bm{\alpha}}) = 0$, which indicates that the all the modes of the unknown target system $\mc S_0$ are retrieved  from the mode dictionary $\Lambda$. The corresponding  $\hat{\bm{\alpha}}$ gives the detetced modes}
\begin{align*}
\hat{\Sigma} = \{ 0.36,  -0.52, 1.02 \}.
\end{align*}
Following \eqref{eq: stack matrix}
we obtain the impulse response  of $\Scal_0$ at time $T$ as    \begin{align}{\label{eq: estimated impulse}}
   	\hat{M}(T) = 
   	\begin{bmatrix}1.02^{T+1} &(-0.52)^{T+1} &
   	0.36^{T+1}
   	\end{bmatrix}
   	\begin{bmatrix}
   		  -0.48\\
   -2.69\\
       3.04
   	\end{bmatrix}.
   \end{align}
   Next we use the reconstructed impulse response data $\hat{M}(T)$ from \eqref{eq: estimated impulse} to compute the LQR controller following the procedures introduced in Section \ref{Section: LQR design}, where the weight matrices are set to be $Q = 4$ and  $R = 1$.   Specifically, we plug  the estimated impulse responses $\hat{M}(1:T)$ in \eqref{eq: Output feedback LQR} to compute the data-driven controller $\subscr{K}{LQR}^0$. We also characterize  the error between $\subscr{K}{LQR}^0$ and $\subscr{K}{LQR}^*$ for
   different data sizes $T$. As depicted in Fig.  \ref{Figure}, we observe that the error between $\subscr{K}{LQR}^0$ and $\subscr{K}{LQR}^*$ decreases exponentially as the number of impulse responses increase, which verifies our results from Section \ref{Section: LQR design} and Section \ref{Section: transfer learning} . 
    \begin{figure}[t]
  \centering
  \includegraphics[width=0.8\columnwidth,trim={0cm 0cm 0cm
    0cm},clip]{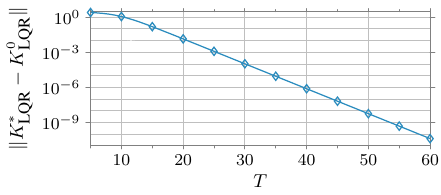}
  \caption{
  {\color{black} This figure shows the error
    $\|\subscr{K}{LQR}^* -    \subscr{K}{LQR}^0 \|$ as a function of the $T$. The error converges exponentially as the number of impulse responses  increases, which implies that  \eqref{eq: Output feedback LQR} reconstruct exactly the output feedback LQR gain}.
    }
  \label{Figure} 
  \end{figure}

\subsection{Modes are not fully contained in the dictionary:  $\Sigma \not \subset \Lambda$}
\textcolor{black}
{
We next consider the scenario in which the modes of the target system $\mc S_0$ are not fully contained within the mode dictionary $\Lambda$, that is, $\Sigma \not\subset \Lambda$. We first collect impulse response $M(1:4)$ from six distinct target systems and apply Algorithm \ref{Al: mode detection} to estimate the modes for each system. \eqref{estimation} returns the elements in the mode dictionary that are closest to the true modes of the target system, and the value of $Z(\hat{\bm\alpha})$ quantifies the discrepancy between the true and estimated modes.}

\textcolor{black}
{Following the procedures in subsection \ref{subsection A}, we characterize the control performance by evaluating the error between $\subscr{K}{LQR}^0$ and $\subscr{K}{LQR}^*$ for each target system. As illustrated in Fig. \ref{Fig2}, target systems associated with smaller values of $Z(\hat{\bm\alpha})$ exhibit notably lower errors compared to those with larger values of $Z(\hat{\bm\alpha})$. This observation is consistent with the fact that a smaller discrepancy between the source and target systems leads to more accurate controller synthesis.
}
 \begin{figure}[t]
  \centering
  \includegraphics[width=1\columnwidth,trim={0cm 0cm 0cm
    0cm},clip]{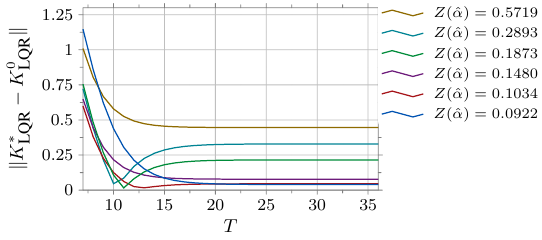}
  \caption{
  {\color{black} This figure shows the error
    $\|\subscr{K}{LQR}^* -    \subscr{K}{LQR}^0 \|$  for unknown target systems with different values of $Z(\hat{\bm \alpha})$ . The error is small for the target systems associated with  small values of $Z(\hat{\bm \alpha})$}.
    }
  \label{Fig2}
\end{figure}

\section{Conclusion}{\label{conclusion}}
In this paper, we study a transfer learning framework for LQR control, where impulse responses are used to learn a data-based LQR controller. We show how the LQR controller can be computed from sufficient impulse response data. Further, by leveraging data from source systems we show that a transfer  mode set can be reused across systems with different dynamics, thus significantly reducing the length of impulse responses for controller design. 
\textcolor{black}{
Our future research includes
(i) characterizing how discrepancies between the target system and its source systems affect closed-loop performance. 
(ii) Identifying conditions under which source system data fail to improve controller synthesis for the target system.
(iii) Generalizing the transfer-learning framework to nonlinear dynamics and to control problems with more general cost functions. 
}
\section{Appendix}
\subsection{Proof of Theorem \ref{thm: Output LQR controller} }
Let $C_t= \begin{bmatrix}
	C & (CA)^{\T} & \cdots & (CA^{N-t})^{\T}
\end{bmatrix}^{\T}$, the batch-form solution  \cite{GS-RES:00}  of the Ricatti equation in
\eqref{eq: Riccati Pt}  is
\begin{align}{\label{batch form}}
	P_{t} = C_t^{\T}(Q_t^{-1}+S_tR_t^{-1}S_t^{\T})C_t.
\end{align}
By substituting \eqref{batch form} into \eqref{eq:Kt} we obtain
{{
\begin{equation}
\begin{aligned}{\label{eq: proof of theorem 2}}
	\subscr{K}{LQR}^{t} = & - 
	\left [ R+\bm{M}_{t+1}^\T (\bm{Q}_{t+1}^{-1}+ \bm{S}_{t+1} \bm{R}_{t+1}^{-1} \bm{S}^{\T}_{t+1})\bm{M}_{t+1} \right ]^{-1} \\
	& \bm{M}^{\T}_{t+1}(\bm{Q}_{t+1}^{-1}+\bm{S}_t\bm{R}_{t+1}^{-1}\bm{S}^{\T}_{t+1})
	\bm{C}_{t+1} A.
 \end{aligned}
 \end{equation}
 }}
 Then by plugging \eqref{eq: proof of theorem 2} into \eqref{eq: static controller}
we obtain $E$ and $F$ in \eqref{eq: Output feedback LQR} as:
	\begin{align*}
	E  & = \bm{C}_{t+1}A F_u \\
	& = 
	\begin{bmatrix}
		CA \\  \vdots\\CA^{T-t+1}
	\end{bmatrix}
	\begin{bmatrix}
		A^{n-1}B & \cdots  &B
	\end{bmatrix} \\
	& =
	\begin{bmatrix}
		M(n+1) & \cdots & M(2) \\
		\vdots & \ddots & \vdots \\
		M(T-t+n+2) & \cdots & M(T-t+3)
	\end{bmatrix},
	\end{align*}
	and $F$ 
	\begin{align*}
	F &= \bm{C}_{t+1} A^{n} \mathcal{O}^{\dagger} \\
	& = 
	\begin{bmatrix}
		CA^{n+1}\mathcal{O}^{\dagger}
 \\ \vdots \\
		CA^{T-t+n+2}\mathcal{O}^{\dagger}
	\end{bmatrix} \\
	& = 
	\begin{bmatrix}
		F_1 \\
		\vdots \\
		F_{T-t+2}
	\end{bmatrix}.
		\end{align*}
		Let $\tilde{K} = \mathcal{O}\mathcal{O}^{\dagger}$ and 
		from \eqref{eq: Calay th} we have $
	 A^n = -\bm{\alpha} 
		\begin{bmatrix}
			I & A^{\T} & \cdots & {A^{n-1}}^{\T}
		\end{bmatrix}^{\T}$.  When $i\geq 0$,  we have 
		\begin{align}
		F_i =CA^{i+n}\mathcal{O}^{\dagger} = -\bm{\alpha} \begin{bmatrix}
			CA^{i} \mathcal{O}^{\dagger} \\ 
			\vdots \\
			CA^{i+n-1} \mathcal{O}^{\dagger}
		\end{bmatrix} = 
		-\bm{\alpha}
		\begin{bmatrix}
			F_{i-n} \\
			\vdots \\
			F_{i-1}
		\end{bmatrix}
		\end{align}		
		and for $ i <0 $ we have $F_i = CA^{i+n} \mathcal{O}^{\dagger} = \tilde{K}_{i+1}$ because 
		\begin{align}
		\tilde{K} = \begin{bmatrix}
			\tilde{K}_1 \\
			\vdots \\
			 \tilde{K}_n
		\end{bmatrix}= \mathcal{O}\mathcal{O}^{\dagger} =
		\begin{bmatrix}
			C\mathcal{O}^{\dagger} \\ \vdots \\ CA^{n-1}\mathcal{O}^{\dagger}
		\end{bmatrix}.
		\end{align} 
		
According to Assumption \ref{Assumpt: controllarbility},   we have $\mathcal{O}\mathcal{O}^{\dagger} = I_n$ for $l=1$, and when $l>1$ we compute $L$ using the QR decomposition of  the data matrix in \eqref{eq: QR decomposition}, where the first $n$ columns of the orthogonal matrix $L$, denoted by $\tilde{\mathcal{O}}$, has the same column space as $\mathcal{O}$. That is, there exists an invertible matrix $\mathcal{D}$ such that $\tilde{\mathcal{O}} = \mathcal{O} \mathcal{D}$, and $\tilde{\Ocal}^{\dagger} = \Ocal^{\dagger}{\Dcal}^{-1}$ because $\Ocal$ is of full column rank.
		Then we can compute $\tilde{L} = \tilde{\mathcal{O}}\tilde{\mathcal{O}}^{\dagger} = {\mathcal{O}} \Dcal \Dcal^{-1}{\mathcal{O}^{\dagger}} = \Ocal \Ocal^{\dagger}$.
		\subsection{Proof of Lemma \ref{lemma: convergence ana}} 
		
		For a given $T$, let $K_t$ and $K^*$ be as in \eqref{eq:Kt} and \eqref{eq:K}, receptively. 
		From \cite{FC-GB-FP:22}  we have 
		$\|K_{T-t} - K^{*}\| \leq c_1 \mu^{t} $ for $0 < \mu <1 $ and $c_1 \in \mathbb{R}^{+}$,  and
		{\footnotesize{
		 \begin{align}
		 \subscr{K}{LQR}^{T-t} - \subscr{K}{LQR}^{*} = (K_{T-t} - K^{*})
		 \underbrace{
		\begin{bmatrix}
		Fu-A^nO^{\dagger}\bm{S}_{T-n+1} & A^nO^{\dagger}
	\end{bmatrix}}_{H}.
	 \end{align}}}
	 Then we have $ \| \subscr{K}{LQR}^{T-t} - \subscr{K}{LQR}^{*} \| \leq  \|(K_{T-t}- K^{*}) \| \|H\| \leq c_1c_2 \mu^t$, where $c2 = \|H\|$.
\subsection{Proof of Lemma \ref{Lemma: Decompose Markov para}} 
Under Assumption \ref{Assumpt: controllarbility} there should exist $n$ distinct modes in $\Sigma$ and $A$ should satisfy $\Tcal^{-1}A\Tcal = \emph{diag}(\sigma_1,\cdots,\sigma_n)$, where $\Tcal \in \mathbb{R}^{n \times n}$ is an invertible matrix. Then, 
\begin{equation*}
{\footnotesize{
\begin{aligned}
	M_i(t) & = \begin{bmatrix}
	\tilde{c}_1 \\ \vdots \\ \tilde{c}_l
\end{bmatrix}
\begin{bmatrix}
	\sigma_1
	^{t-1}  \\ 
	    &    \ddots   \\
    & &   \sigma_n^{t-1}
\end{bmatrix} \\
\tilde{b}_i
& = 
 \begin{bmatrix}
 \begin{bmatrix}
 	\sigma_1^{t-1} & \cdots & \sigma_n^{t-1}
 \end{bmatrix} 
\emph{diag}(\tilde{c}_{1})\tilde{b}_i \\
\vdots \\
 \begin{bmatrix}
 	\sigma_1^{t-1} & \cdots & \sigma_n^{t-1}
 \end{bmatrix} 
\emph{diag}(\tilde{c}_{l})\tilde{b}_i
 \end{bmatrix} \\
 & =
 \begin{bmatrix}
 	\sigma_1^{t-1} & \cdots & \sigma_n^{t-1} \\
& &  &\ddots \\
 & & & & 	\sigma_1^{t-1} & \cdots & \sigma_n^{t-1}
 \end{bmatrix} 
 \underbrace{
  \begin{bmatrix}
  \emph{diag}(\tilde{c}_{1}){b}_i
  \\
  \vdots \\
  \emph{diag}(\tilde{c}_{l}){b}_i
  \end{bmatrix}}_{W_i}.
  \end{aligned}}}
\end{equation*}
This leads to \eqref{eq: Decompose Markov para}.
\subsection{Proof of Lemma \ref{lemma: impulse response and modes}}
Using  Cayley-Hamilton Theorem we have 
\begin{align}{\label{eq: lemma 1}}
A^n + I \alpha_1 +  A \alpha_2 +\cdots+ A^{n-1} \alpha_n= 0.
\end{align}
By multiplying $C$ to the left hand side of \eqref{eq: lemma 1} and multiplying $B$ to the right hand side of \eqref{eq: lemma 1} we obtain 
\begin{align*}
CA^nB + C B \alpha_1 +  CAB \alpha_2 +\cdots+ CA^{n-1}B \alpha_n  = 0,
\end{align*}
which is equivalent to 
\begin{align*}
	M(n+1)+ M(1) \alpha_1 + M(2) \alpha_2+\cdots+M(n)\alpha_n=0.
\end{align*}
This leads to \eqref{eq: M and Sigma}. 

\section*{References and Footnotes}

\bibliographystyle{unsrt}
\bibliography{alias,Main,FP,New}
\end{document}

%% file: preambleLettersDef.tex
\newcommand{\bv}[1]{\mathbf{#1}}		

\newcommand{\jbar}{\bar{j \phantom{\tiny \,}} \kern -0.1em}

\newcommand{\jhat}{\hat{j \phantom{\tiny \,}} \kern -0.1em}

\newcommand{\jtil}{\tilde{j \phantom{\tiny \,}} \kern -0.1em}

\newcommand{\vj}{\bv{j}}

\newcommand{\vjbar}{\bar{\vj \phantom{\tiny \,}} \kern -0.1em}

\newcommand{\vjhat}{\hat{\vj \phantom{\tiny \,}} \kern -0.1em}

\newcommand{\vjtil}{\tilde{\vj \phantom{\tiny \,}} \kern -0.1em}

\newcommand{\Acal}{\mathcal{A}}

\newcommand{\Dcal}{\mathcal{D}}

\newcommand{\Ocal}{\mathcal{O}}

\newcommand{\Scal}{\mathcal{S}}
\newcommand{\Tcal}{\mathcal{T}}

\usepackage{bm}

\newcommand{\vAcalhat}{\hat{\bm{\Acal} \phantom{a}} \kern-0.5em}
\newcommand{\vAcalcheck}{\check{\bm{\Acal} \phantom{a}} \kern-0.5em}
\newcommand{\vAcalbar}{\bar{\bm{\Acal} \phantom{a}} \kern-0.5em}

\usepackage{upgreek}

